\documentclass[runningheads]{llncs}
\usepackage[utf8]{inputenc}
\usepackage[T1]{fontenc}
\usepackage{tcolorbox}
\usepackage{amsmath,amsfonts,amssymb}
\usepackage{bm}
\usepackage{subfigure}

\usepackage[colorinlistoftodos,prependcaption,textsize=tiny,]{todonotes}

\usepackage{pifont}
\usepackage{algorithm}
\usepackage[noend]{algorithmic}
\newcommand{\np}{{NP}}

\newcommand{\fpt}{{FPT}}

\newcommand{\mmbcs}{{\textsc{Balanced Connected Subgraph}}}
\newcommand{\old}[1]{{}}
\newcommand{\bcs}{\textsc{BCS}}
\newcommand{\rst}{{\textit{RST}}}

\usepackage{graphicx}

\newcommand{\colg}[1]{{\color{green!30!black}{{\textit{#1}}}}}
\newcommand{\colb}[1]{{\color{blue}{{\textit{#1}}}}}
\newcommand{\colr}[1]{{\color{red}{{\textit{#1}}}}}
\usepackage{tcolorbox}
\usepackage{xcolor}

\usepackage[colorlinks=true,urlcolor=blue,citecolor=blue,linkcolor=red!50!black,bookmarks=true]{hyperref}

\begin{document}
\title{Balanced Connected Subgraph Problem in Geometric Intersection Graphs}
%
%
\author{Sujoy Bhore\inst{1}\and
Satyabrata Jana\inst {2}\and 
Supantha Pandit\inst{3}\and
Sasanka Roy\inst{2}}
\authorrunning{S. Bhore et al.}
%
\institute{Algorithms and Complexity Group, TU Wien, Vienna, Austria \and Indian Statistical Institute, Kolkata, India \and
Dhirubhai Ambani Institute of Information and Communication Technology, Gandhinagar, Gujarat, India
\\ \email{\{sujoy.bhore, satyamtma, pantha.pandit, sasanka.ro\}@gmail.com}}

\maketitle              
%
\begin{abstract}
We study the \colg{\mmbcs}~(shortly, \colg{\textbf{\bcs}}) problem on geometric intersection graphs such as interval, circular-arc, permutation, unit-disk, outer-string graphs, etc. Given a \colg{vertex-colored} graph $G=(V,E)$, where each vertex in $V$ is colored with either ``\colr{red}'' or ``\colb{blue}'', the \bcs~problem seeks a maximum cardinality induced connected subgraph $H$ of $G$ 
such that $H$ is \colg{color-balanced}, i.e., $H$ contains an equal number of red and blue vertices. We study the computational complexity landscape of the \bcs~problem while considering geometric intersection graphs. On one hand, we prove that the \bcs~problem is \np-hard on the unit disk, outer-string, complete grid, and unit square graphs. On the other hand, we design polynomial-time algorithms for the \bcs~problem on interval, circular-arc and permutation graphs. In particular, we give algorithm for the \colg{\textsc{Steiner Tree}} problem on both the interval graphs and circular arc graphs, that is used as a subroutine for solving \bcs~problem on same graph classes. Finally, we present a \fpt~algorithm for the \bcs~problem on general graphs. 


\end{abstract}
\section{Introduction}

The \colg{intersection graph} of a collection of sets is a graph where each vertex of the graph represents a set and there is an edge between two vertices if their corresponding sets intersect. Any graph can be represented as an intersection graph over some sets. Geometric intersection graph families consist  intersection graphs for which the underlying collection of sets are some geometric objects. Some of the important graph classes in this family are interval graphs (intervals on real line), circular-arc graphs(arcs on a circle), permutation graphs (line segments with endpoints lying on two parallel lines), unit-disk graphs (unit disks in the Euclidean plane), unit-square graphs (unit squares in the Euclidean plane), outer-string graphs (curves lying inside a disk, with one endpoint on the boundary of the disk), etc.  In the past several decades, geometric intersection graphs became very popular and extensively studied due to their interesting theoretical properties and applicability. 

In this paper, we consider an interesting problem on general vertex-colored graphs called the \colg{\mmbcs} (shortly, \colg{\bcs}) problem.  A subgraph $H=(V',E')$ of $G$ is called \colg{color-balanced} if it contains an equal number of red and blue vertices.

\begin{tcolorbox}[colback=red!5!white]
	\noindent {{\textcolor{red!50!black} {\textbf{\textsc{Balanced Connected Subgraph} (\bcs) Problem} }}}\\
	{\bf Input:} A graph $G=(V,E)$, with node set $V=V_R\cup V_B$ partitioned into red nodes ($V_R$) and blue nodes ($V_B$).\\
	{\bf Output:} Maximum-sized color-balanced induced connected subgraph.
\end{tcolorbox}
\vspace{-3mm}

\subsection{Previous Work}  Previously the \bcs~problem  has been studied  on various graph families such as trees, planar graphs, bipartite graphs, chordal graphs, split graphs, etc \cite{bhore2019balanced}. Most of the findings suggest that the problem is \np-hard for general graph classes, and it is possible to design a polynomial algorithm for restricted classes with involved approaches. In that paper, we have pointed out a connection between the \bcs~problem and  \colg{graph motif} problem (see, e.g., \cite{Bonnet2017,Fellows2011,Lacroix2016}). In the graph motif problem, we are given the input as a graph $G = (V,E)$, a color function $ col: V \rightarrow \mathcal{C}$ on the vertices, and a multiset $M$, called motif, of colors of $\mathcal{C}$; the objective is to find a subset $V'\subseteq V$ such that the induced subgraph on $V'$ is connected and $col(V')=M$. Note that, if $\mathcal{C}=\{$red, blue$\}$ and the motif has the same number of red and blues then, the solution of the graph motif problem gives a balanced connected subgraph. Indeed, a solution to the graph motif problem provides one balanced connected subgraph, with an impact of a polynomial factor in the running time. However, it does not guarantee the maximum size balanced connected subgraph. Nonetheless, the NP-hardness result for the \bcs~problem on any particular graph class implies the NP-hardness result for the graph motif problem on the same class. 
Graph motif problem has wide range of applications in bioinformatics \cite{bodlaender1995intervalizing}, DNA physical mapping \cite{fellows1993dna}, perfect phylogeny \cite{bodlaender1992two}, metabolic network analysis \cite{lacroix2006motif}, protein–protein interaction networks and phylogenetic analysis \cite{bodlaender2007quadratic}.  This problem  was introduced in the context of detecting patterns that occur in  the interaction networks between chemical compounds and/or reactions \cite{lacroix2006motif}. 
 
In this work, we revisit the \bcs~problem on some popular geometric intersection graphs such as interval graphs, circular-arc graphs, permutation graphs, unit-disk graphs, outer-string graphs, unit-square graphs, complete-grid graphs. There are many graph theoretic problems that are \np-hard for general graphs but polynomially solvable while considering geometric intersection graphs. For example, The clique decision problem is \np-complete for general graph \cite{karp1972reducibility}, however polynomially solvable for interval graphs \cite{imai1983finding}, circular-arc graphs \cite{imai1983finding}, permutation graphs \cite{pecher2010clique}, unit-disk graphs \cite{clark1991unit}.
Our hope is to exploit the geometric properties of these 
restricted graph families to achieve theoretical results. 

\subsection{Our Results} 
We present a collection of results on the \bcs~problem on geometric intersection graphs, that advances the study of this problem on diverse graph families.

\begin{itemize}
\item[\ding{229}] On the hardness side, in Section~\ref{hard}, we show that the \bcs~problem is NP-hard on unit disk graphs, outer-string graphs, complete grid graphs, and unit square graphs. 

\item[\ding{229}] On the algorithmic side, in Section~\ref{algo}, we design polynomial-time algorithms for interval graphs ($\mathcal{O}(n^4\log n)$ time),
circular-arc graphs ($\mathcal{O}(n^6\log n)$ time) and 
permutation graphs ($\mathcal{O}(n^6)$).
Moreover, we give an algorithm for the \textsc{Steiner Tree} problem on the interval graphs, that is used as a subroutine in the algorithm of the \bcs~problem for intervals graphs.
Finally, we show that the \bcs~problem is fixed-parameter tractable for general graphs ($ 2^{\mathcal{O}(k)}n^2 \log n $) while parameterized by the number of vertices in a balanced connected subgraph.
\end{itemize}

\section{Hardness Results}\label{hard}

In this section, we consider \bcs~problem on unit disk graphs, unit square graphs, complete grid graphs, and outer-string graphs.


\subsection{Unit Disk Graphs}
In this section, we study the \bcs~problem on unit-disk graphs. 
It has been shown that this problem is \np-hard on planar graphs \cite{bhore2019balanced}. 
Besides, we know that every planar graph can be represented as a disk graph (due to Koebe's kissing disk embedding theorem \cite{dehmer2010structural}). Therefore, the \np-hardness for the \bcs~problem on disk graphs directly follows from there. Here, we show that this problem remains \np-hard even on unit-disk graphs. 

Here we give a reduction from the \colg{\textsc{Rectilinear Steiner Tree} (\rst)} problem \cite{garey1977rectilinear}.
In this problem, we are given a set $P$ of integer coordinate points in the plane and an integer $L$. The objective is to find a Steiner tree $T$ (if one exists) with length at most $L$. 

During the reduction, we first generate a geometric intersection graph from an instance $X(P,L)$ of the \rst~problem. The vertices of the geometric intersection graph are integer coordinated and each edge is of unit length. Next, we show that this geometric intersection graph is a unit-disk graph.


\noindent {\bf Reduction:}
Suppose we have an instance $X(P,L)$. For any point $p\in P$, $p(x)$ and $p(y)$ denote the $x$- and $y$-coordinates of $p$, respectively. Let $p_t $, $ p_b $, $ p_l $, and $ p_r $ be the topmost (largest $y$-coordinate), bottom-most (smallest $y$-coordinate), leftmost (smallest $x$-coordinate), and rightmost (largest $x$-coordinate) points in $ P $. We now take a unit integer rectangular grid graph $D$ on the plane such that the coordinates of the lower-left grid vertex is $(p_l(x),p_b(y))$ and upper-right grid vertex is $(p_r(x),p_t(y))$. Now we associate each point $p$ in $P$ with a grid vertex $d_p$ in $D$ having the same $x$- and $y$-coordinates of $p$. 
Now we assign colors to the points in $D$. The vertices in $D$ correspond to the points in $P$ are colored with red and the remaining grid vertices in $D$ are colored with blue. We now add some more vertices to $D$ as follows:

Observe that if there is a Steiner tree $T$ of length $L +1=|P|$ exists then $T$ does not include any blue vertex in $D$. Further, if there is a Steiner tree $T$ of length $L +1=2|P|$ exists then $T$ contains equal number of red and blue vertices in $D$. Based on this observation we consider two cases to add some more vertices (not necessarily form a grid structure) to $D$.

\begin{description}
\item[Case 1. \boldmath{[$L+1\geq 2|P|$]}]\label{ca1} In this case the number of blue vertices in a Steiner tree $T$ (if exists) is more than or equals to red vertices in $D$.  We consider a path $\delta$ of $(L-2|P|+1)$ red vertices starting and ending with vertices $r_1$ and $r_{L-2|P|+1}$, respectively. The coordinates of $r_i$ is $(p_l(x)-i,p_l(y))$, for $1\leq i\leq L-2|P|+1$.  We connect this path with $D$ using an edge between the vertices $r_1$ and $p_l$. See Figure \ref{fig:udg-nphard-1} for an illustration of this construction. Let the resulting graph be $G_1=D\cup \delta$.

\item[Case 2. \boldmath{[$ L+1 < 2|P|$]}] In this case the number of red vertices in a Steiner tree $T$ (if exists) is more than the number of blue vertices in $D$. We consider a path $\delta$ of $(2|P| -L)$ blue vertices starting and ending with vertices $b_1$ and $b_{2|P| -L}$, respectively. The coordinates of $b_i$ is $(p_l(x)-i,p_l(y))$, for $1\leq i\leq 2|P|-L$. We connect this path with $D$ using an edge between the vertices $b_1$ and $p_l$. We add one more red vertex $r'$ whose coordinates are $(p_{2|P|-L}(x)-1,p_l(y))$ and connect it with $b_{2|P|-L}$ using an edge. See Figure \ref{fig:udg-nphard-2} for an illustration of this construction. Let the resulting graph be $G_2=D\cup \delta \cup \{r'\}$
\end{description}

\vspace{-.5cm}
\begin{figure}[ht!] 
\begin{center}
\subfigure[ ]{\includegraphics[scale=.63]{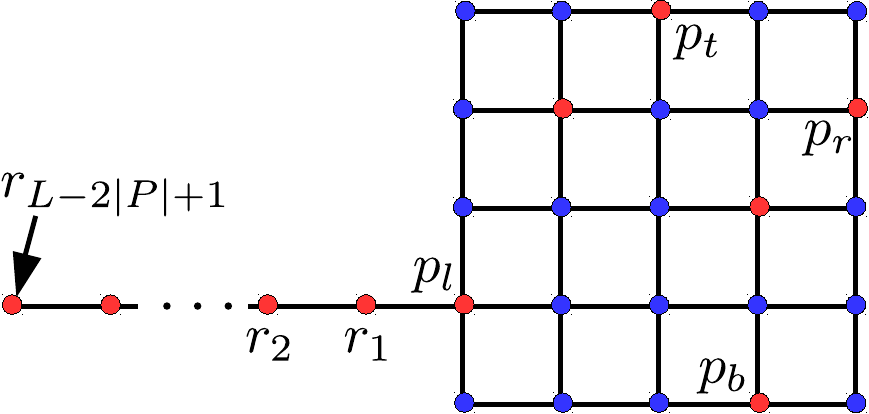}
\label{fig:udg-nphard-1}
}
\hspace{.5cm}
\subfigure[ ]{\includegraphics[scale=.63]{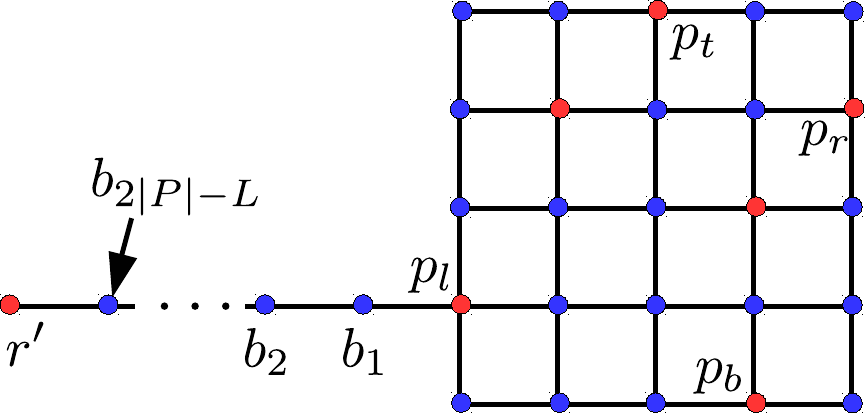}
\label{fig:udg-nphard-2}
}

\end{center}
\caption{(a) Construction of the instance $G_1$. (b) Construction of the instance $G_2$.}
\label{fig:udg-nphard}
\end{figure}
\vspace{-.5cm}

This completes the construction. Clearly, the construction (either $G_1$ or $G_2$) can be done in polynomial time. Now we prove the following lemma for Cases 1 and 2 separately.  



\begin{lemma}\label{lem-bcs-np-hard}
	The instance $X$ of the \rst~problem has a solution $ T $ if and only if 
	\begin{itemize}
	    \item {\bf For Case 1:} the instance $G_1$ has a balanced connected  subgraph $H$ with $2(L-|P|+1)$ vertices.
	    \item {\bf For Case 2:} the instance $G_2$ has a balanced connected  subgraph $T$ with  $2(|P|+1)$ vertices.
	\end{itemize}
\end{lemma}

\begin{proof} We prove this lemma for Case 1. The proof of Case 2 is similar.

\noindent {\bf For Case 1:}
Assume that $X$ has a Steiner tree $T$ of length $L$, where $L+1\geq 2|P|$. Let $ U $ be the set of those  vertices in $ G_1 $ corresponds to the vertices in $T$. Clearly, $U$ contains $|P|$ red vertices and $L-|P|+1$ blue vertices. Since $L+1\geq 2|P|$, to make $U$ balanced it needs $L-|P|+1 -|P|$ more red vertices. So we can add the path $\delta$ of $L-2|P|+1$ red vertices to $U$. Therefore, $U\cup\delta$ becomes connected and balanced (contains $L-|P|+1$ vertices in each color).

On the other hand, assume that there is a balanced connected subgraph $H$ in $G$ with $(L-|P|+1)$ vertices of each color. We can observe that $H$ is a tree and no blue vertex in $H$ is a leaf vertex. The number of red vertices in $ G_1$ is exactly $(L-|P|+1)$. So the $H$ must pick all the $(L-|P|+1)$  blue vertices that connect the vertices in $G_1$ corresponding to $ P $.
We take the set $A$ of all the grid vertices corresponding to the vertices in $ H $ except the vertices $ \{r_i; 1 \leq i \leq (L-2|P|+1)\} $. We output the Steiner tree $T$ that contains the vertex set $A$ and edge set $E_A$ connecting the vertices of $A$ according to the edges in $H$. As $ |A|= 2(L-|P|+1)- (L-2|P|+1)=L+1 $, so we output a  Steiner tree of length $ L $. \qed
\end{proof}

 We now show that the geometric intersection graph either $G_1$ or $G_2$ is a unit-disk graph. Let us consider the graph $G_1$. For each vertex $v$ in $G_1$ we take a unit disk whose radius is $\frac{1}{2}$ and center is on the vertex $v$. Therefore from the Lemma \ref{lem-bcs-np-hard}, we conclude that,
 
 \begin{theorem}
The \bcs~problem is \np-hard for unit-disk graphs.
 \end{theorem}

\noindent {\bf Extensions:} By a simple extension to the the above reduction we can prove that The \bcs~problem is \np-hard for the unit square graph and the complete grid graph.

\subsection{Unit Square Graphs} \label{app-sq}
 
We show that the \bcs~problem remains \np-hard for the unit square graphs. Similar to the unit disk graph, we give a reduction from the \rst~problem. The reduction (construction of the graph $G_1$ or $G_2$ and the proof similar to Lemma \ref{lem-bcs-np-hard}) is exactly same as the reduction for unit disk graphs. The only thing we show that both the graph $G_1$ and $G_2$ are intersection graph of unit squares.

\begin{figure}[h!]
	\centering
	\includegraphics[scale=0.65]{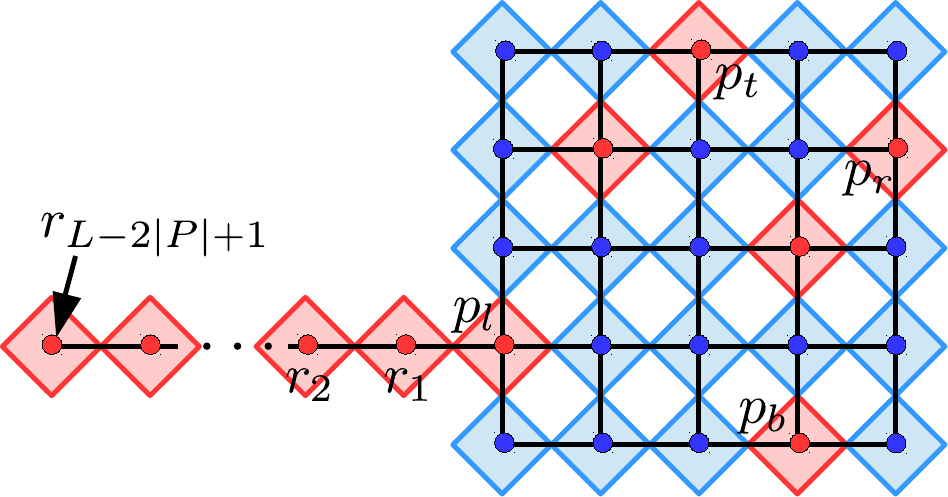}
	\caption{Construction of the instance for \bcs~problem for unit-disk graph}
	\label{fig-np-sq}
\end{figure} 

Let us consider the graph $G_1$. For each grid vertex $v$ in $G_1$, take an axis-parallel square and rotate it $45^\circ$ with the $x$-axis. The side length of this square is $\frac{1}{\sqrt{2}}$ and whose center is on $v$.  See Figure \ref{fig-np-sq} for an illustration. Finally, we rotate the complete construction by an angle of $-45^\circ$ and scale it by a factor of $\sqrt{2}$. This makes the squares axis-parallel and unit side length. It is not hard to verify that $G_1$ is an intersection graph of this set of squares. Hence we conclude:

\begin{theorem}
	The \bcs~problem is \np-hard for axis-parallel unit square graphs.
\end{theorem}
 
\subsection{Complete Grid Graphs} \label{app-cgd} 

We show that the \bcs~problem is \np-hard for the complete grid graphs. 
Notice that to prove that the \bcs~problem is \np-hard for the unit-disk graph we generate an instance either $G_1$ or $G_2$. None of these graphs is a complete grid graph. Our idea is to make them complete grid graphs. Now consider the graph $G_1$. Notice that $G_1$ has two components $D$ and $\delta$. $D$ is a complete grid graph. We add $\delta$ to it to the left of $D$. We now add blue vertices at each integer coordinates  $(r,s)$, where $(r,s) \notin \delta,( p_l(x)-L +2|P|-1) \leq r < p_l(x) $ and $ p_b(y) \leq s \leq p_t(y)$. As a result $G_1$ becomes a complete grid graph. Similarly, we can make $G_2$ complete grid graph by adding only blue vertices. For this modified $G_1$ and $G_2$, we can prove Lemma \ref{lem-bcs-np-hard} using similar arguments. Hence we conclude:


\begin{theorem}
	The \bcs~problem is \np-hard for complete grid graphs.
\end{theorem}

\subsection{Outer-String Graphs}
We study the \bcs~problem on string graphs. We know finding a balanced subgraph on planar graphs is NP-hard. Now, every planar graph can be represented as a string graph( by drawing a string for each vertex that loops around the vertex and around the midpoint of each adjacent edge). So \np-hardness of \bcs~problem  holds for string graphs. Below we show that this problem remains \np-hard even for outer-string graphs.

\vspace{2mm}
 We give a reduction from the dominating set problem which is known to be \np~complete on general graphs \cite{kikuno1980np}. Given a graph $ G=(V,E) $, the dominating set problem asks whether there exists a set  $ U \subseteq V $ such that $|U|\leq k$ and $ N[U]=V $, where $ N[U] $ denotes neighbours of $ U $ in $ G $ including $ U $ itself.

During the reduction, we first generate a geometric intersection  graph $H=(R\cup B,E')$ from an instance $X(G,k)$ of the dominating set problem on general graph. Next, we show that $H$ is an outer-string graph.
\vspace{2mm}

\noindent {\bf Reduction:}
Let $ G=(V,E) $ be  graph with vertex set $ V=\{v_1, v_2,\dots, v_n\}$. For each vertex $ v_i \in V $ we add a red vertex $ v_i $ and a blue vertex $ v'_i $ in $ H $. For each edge $ (v_i, v_j) \in E $, we add two edges $(v_i, v'_j), (v'_i, v_j)  $ in $ E' $. Take a path of $ k $ red vertices starting at $ r_1 $ and ending at $ r_k $. Also take a path of $ n $ blue vertices starting at $ b_1 $ and ending at $ b_n$. Add the edges $ (b_n, r_k) , (b_1,v_1)$ into $ E' $.  We add edges between all pair of vertices in $\{ v'_1,v'_2, \dots v'_n\}$. Our construction ends with adding $ n $ edges $ (v_i, v'_i) $ into $ E' $ for $ 1 \leq i \leq n $.

\begin{figure}[ht!]
\centering
\subfigure[]
		{
			\includegraphics[scale=0.5]{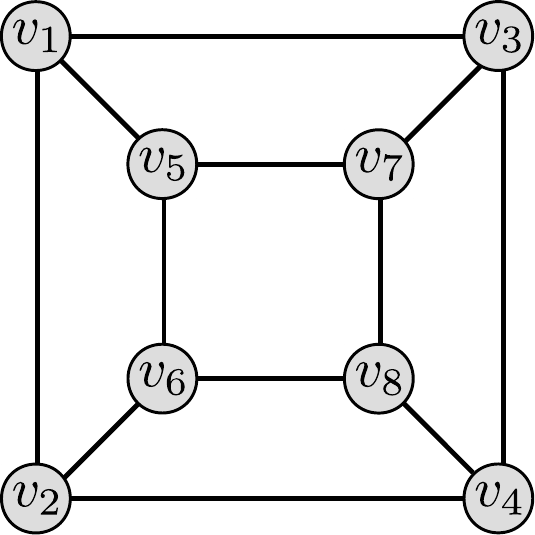}			\label{fig:subfig1}
		}
		\subfigure[]
		{
			\includegraphics[scale=0.5]{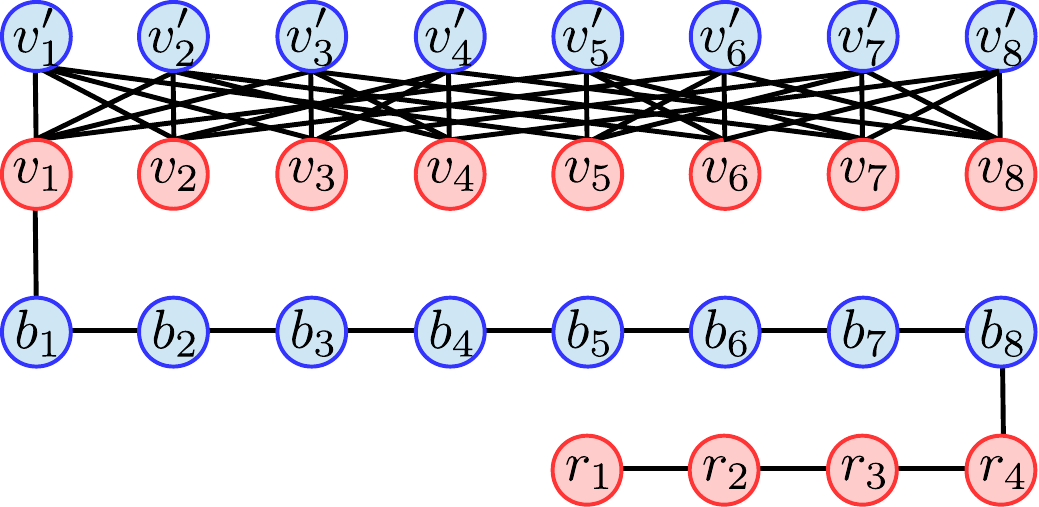}
			\label{fig:subfig11}
		}

		\subfigure[]
		{
			\includegraphics[scale=1]{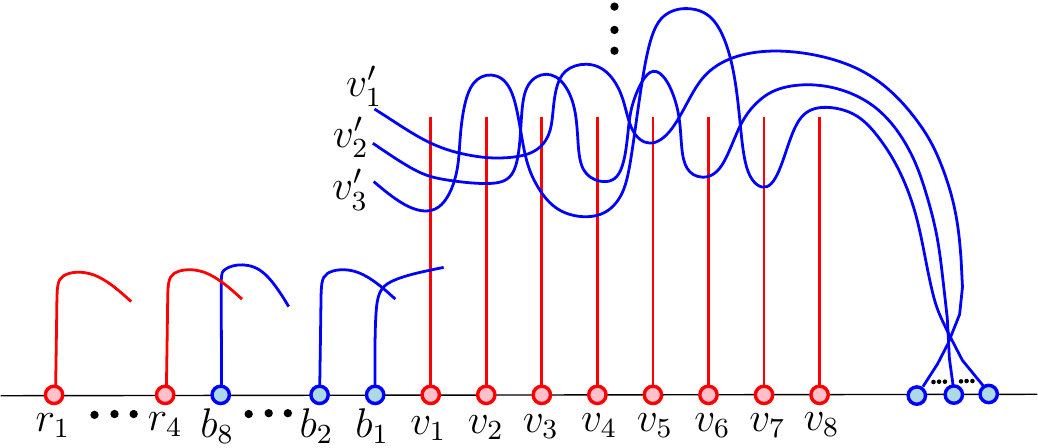}
			\label{fig:subfig12}
		}
		\caption[(Optional caption for list of figures]{(a) a graph $ G $. (b) Construction of $ H $ from $ G $ with $ k=4 $. For the clarity of the figure we omit the edges between each pair of vertices $v'_i$ and $v'_j$, for $i \neq j$. (c) Intersection model of $ H $. }
		\label{fig_out}
	\end{figure}

This completes the construction.  See Figure \ref{fig_out} for an illustration of this construction. Clearly, the construction can be made in polynomial time. Now we prove the following lemma. 
	
\begin{lemma}\label{lem-bcs-np-hard-string}
The instance $X$ has a dominating set of size $k$ if and only if $ H $ has a balanced connected subgraph $T$ with $2(n+k)$ vertices.
\end{lemma}
	
\begin{proof}
Assume that $G $ has a dominating set $ U $ of size $ k $.  Now we take the subgraph $ T $ of $ H $ induced by $\{v'_i \colon v_i \in U  \} $ along with the vertices $\{r_i \colon 1\leq i \leq k \} \cup \{b_j \colon 1 \leq j \leq n\} \cup \{v_i\colon 1 \leq j \leq n\}$ in $ H $. Now clearly $ H $ is connected and balanced with $ 2(n+k) $ vertices.
		
On the other hand, assume that there is a balanced connected tree $T$ in $H$ with $(n+k)$ vertices of each color. The number of red vertices in $ H$ is exactly $(n+k)$. So the solution must pick the  blue vertices $ \{b_i \colon 1 \leq i \leq n \} $ that connect $v_1 $ with $ r_1 $. As $ T $ has exactly $ (n+k) $ blue vertices then it $ T $ should pick exactly $ k $ vertices from the set $ \{v'_i \colon 1 \leq i \leq n \} $. The set of vertices in $V$ corresponding to these $k$ vertices gives us a dominating set of size $ k $ in $ G $. \qed		
\end{proof}

We now verify that $ H $ is an outer-string graph. For this, we provide an intersection model of it consisting of curves that lie in a common half-plane.  For an illustration see Figure \ref{fig:subfig12}. We draw a horizontal line $ y=0 $. For each vertex $ v_i \in H $, draw the line segment $L_i= \overline{(i,0)(i,1)} $. For each vertex $ v'_i \in H $, we draw a curve $ C_i $, having one endpoint on the line $y=0$, such a way that it touches only the lines $ L_i \cup L_j $ where $ (v_i, v_j) \in E $. Now we can easily add the curves corresponding to    $\{r_i \colon 1\leq i \leq k \} \cup \{b_j \colon 1 \leq j \leq n\}$ having one endpoint on the line $ y=0 $ with satisfying the adjacency.

Finally, using Lemma \ref{lem-bcs-np-hard-string}, we conclude that,

\begin{theorem}
The \bcs~problem is \np-hard for the outer-string graphs.
\end{theorem}

\section{Algorithmic Results}\label{algo}
In this section, we consider \bcs~problem on interval graphs, circular-arc graphs, and permutation graphs.
\subsection{Interval Graphs}
In this section, we study the \mmbcs~problem on connected interval graphs. A graph $G=(V,E)$ is called an interval graph if each vertex $u \in V$ is associated with an interval ${I_u= [l_u,r_u]}$ (where $l_u$ and $r_u$ denote the left and right endpoint of $I_u$, respectively), and for any pair of vertices $u,v\in V$, $(u,v) \in  E(G)$ if and only if ${I_u } \cap {I_v } \ne \phi$. Given an $n$ vertex interval graph $G=(V,E)$, we order the vertices of $G$ based on the left endpoints of their corresponding intervals. Consider a pair of vertices $u,v\in V$ with $l_u \le l_v$, we define a set $S_{u,v}= \{w \colon w \in V, l_u \leq l_w <r_w \leq r_v\} \cup \{u,v\}$. We also consider the case when $u=v$, and define $S_{u,u}$ in a similar fashion. Let $H$ be a subgraph of $G$ induced by $S_{u,v}$ (resp. $S_{u,u}$) in $G$. For any $u,v\in V$, consider $S_{u,v}$ and let $r$ and $b$ be the number of red and blue vertices in $S_{u,v}$, respectively. Without loss of generality, we may assume that $b\le r$. The goal is to find a \bcs~in $ H $ with cardinality $ |V(H)|=2b $ and $ u,v \in V(H) $ (if exists). Let $B$ be the set of all blue intervals and $T=B\cup \{u,v\}$ (resp. $T=B\cup \{u\}$, if $u=v$).

First, we compute  a connected subgraph that contains $T$ and some extra red intervals. In order to do so, we require an algorithm for Steiner tree problem on interval graphs. 

\paragraph{Steiner Tree on Interval Graphs:} 
Given an simple connected interval graph $G=(V,E)$ and the terminal set $T\subseteq V$, the minimum Steiner tree problem seeks a smallest tree that spans over $T$. 
The remaining vertices $S=V\setminus T$ are denoted as Steiner vertices. We describe a simple greedy algorithm that computes minimum Steiner tree on $G$.

We first break the terminal set $T$ into $m$ (for some $m\in [n]$) components $\{C_1,\ldots,C_m\}$
sorted from left to right based on the right endpoints.
We assume that $\{I_{C_1},\ldots,I_{C_m}\}$ be the rightmost intervals of these components. 
We consider the first component $C_1$ and the neighborhood set of $I_{C_1}$, i.e., $N(I_{C_1})$.
Let $I_j$ be the rightmost interval in $N(I_{C_1})$. By rightmost we mean the interval having rightmost endpoint. We add $I_j$ in the solution $D$.
Now, we recompute the components based on $T\cup D$. Note, $C_1\cup I_j$ is now contained in one component.  We repeat this produce until $T\cup D$ becomes one single component. 
The pseudo-code of this procedure is given below.  

\vspace{-.3cm}
\floatname{algorithm}{Algorithm}
\begin{algorithm}[ht]
\caption{$Select \_ Steiners(G=(V,E),T)$ } \label{algo:St-int}
\begin{algorithmic}[1]
\STATE $S\leftarrow V\setminus T$
\STATE $D\leftarrow 0$

\STATE $C\leftarrow$ set of components induced by $T \cup D$

\STATE Let $I_{C_1}, \ldots, I_{C_m}$ be the left to right ordered set of the right endpoints of the components. Where $C=\{C_1, C_2, \ldots C_m\}$ be a set of $m$ components (for some $m \geq 1$).

\STATE $D\leftarrow D\cup {I_j}$
\\where $I_j \in S$ and the rightmost interval in $N(I_{C_1})$.
\STATE go to step 3.
\STATE Repeat until $C$ consists of one component.

\RETURN $D$
\end{algorithmic}
\end{algorithm}

\vspace{-.5cm}

\paragraph{Correctness:} We prove the correctness of Algorithm~\ref{algo:St-int} in two steps. First, we show that the algorithm returns a solution set $D\in V\setminus T$ such that the graph induced by $T\cup D$ is connected. Then, we prove that the solution $D$ is optimum. First we show that Algorithm~\ref{algo:St-int} produce a component  connecting $T$. Assume that the algorithm has not produced a connected component containing all the intervals of $T$, it means that we did not reach the final step of the algorithm (see line 7 in Algorithm~\ref{algo:St-int}). So, there must exists an interval $I_i \notin C_m$ (for $i\in [m]$) such that no interval in $G$ cross the right endpoint of $I_i$, otherwise, we would have picked an interval from its neighborhood in the algorithm. It means $G$ is not a connected graph. However, we began with a connected interval graph. Thus, our assumption is wrong, and the graph induced by $T\cup D$ is connected.


Now we prove the optimality of the Algorithm~\ref{algo:St-int} by induction. The base case is that we have to connect the first two components $C_1$ and $C_2$. We choose the rightmost interval from the neighborhood of $I_{C_1}$. Note, that we have to connect $C_1$ and therefore it is inevitable that we have to pick an interval from $N(I_{C_1})$. We choose
the rightmost interval (say, $(I_\ell)$).  
Now, if this choice already connects $C_2$, we are done. Otherwise, $C_1=C_1\cup I_{\ell}$. Now, let us assume that we have obtained an optimal solution until step~$i$. At step~$i+1$, we have to connect the first two components. By applying the same argument as the base case, we choose the rightmost interval from the first component and proceed. 
This completes the proof. 

\paragraph{Time Complexity:} Our algorithm runs in $\mathcal{O}(n^2)$ time. 
First, in Step~$1$ we break the terminal set $T$ into a set of $m$ (for some $m\in [n]$) components, which takes linear time on $n$. Now, in constant time we can choose the rightmost interval $I_\ell$ from $N(I_{C_1})$.
If $I_\ell$ connects other components (say until $C_j$, for $j\in [m]$) then $C_1=\bigcup_{i=1}^{j} C_i$. Otherwise, $C_1=C_1\cup I_\ell$.
This process takes linear time as well. Therefore, in total our algorithm returns a Steiner in quadratic time. 

Now, we go back to the \bcs~problem.
Recall that, for any pair of intervals $u,v\in V$, our objective is to compute a \bcs~with vertex set $T$ of cardinality $2b$ (if exists), where $b$ and $r$ is the number of red and blue intervals in $S_{u,v}$, respectively, and $b\leq r$. Let $H$ be the  subgraph of $G$ induced by $S_{u,v}$.
Let $R$ and $B$ denote the set of red and blue intervals in $S_{u,v}$, respectively. 
We describe this process in Algorithm~\ref{algo:bcs-int}.
We repeat this procedure for every pair of intervals and report the solution set with the maximum number of intervals. 
\vspace{-.3cm}

\floatname{algorithm}{Algorithm}
\begin{algorithm}[ht]
\caption{$BCS\_Interval(H)$} \label{algo:bcs-int}
\begin{algorithmic}[1]
\STATE $T\leftarrow B \cup \{u,v\}$
\STATE $D=Select \_ Steiners(H,T)$
\STATE $r'\leftarrow $ number of red vertices in $ D \cup \{u,v\} $
\STATE $b'\leftarrow $ number of blue vertices in $T $
\IF{$r'>b'$}
\STATE Return $ \phi $ 
\ENDIF
\IF{$ r'=b' $}
\STATE Return $ G[D \cup T] $
\ENDIF
\IF{$ r'<b' $}
\STATE Return $ G[D \cup T \cup X] $ \\where $ X \subset V(H)$ is the set of $ (b'-r') $ red vertices with $ X \cap (D\cup T)= \phi$.
\ENDIF
\end{algorithmic}
\end{algorithm}
\vspace{-.5cm}
\paragraph{Correctness:} 
We prove that our algorithm yields an optimum solution.
Let $G'$ be such a solution. 
Let $u$ and $v$ be the intervals with leftmost endpoint and rightmost endpoint of $G'$, respectively.
Now we show  that $V(G')= \min \{2r,2b\}$, where $r$ and $b$ be the number of red and blue color vertices is $S_{u,v}$, respectively. 
Let us assume $V(G') \neq \min \{2r,2b\}$. Then there exists at least one blue interval $z$ and one red interval $z'$ that belong to $S_{u,v} \setminus V(G')$. 
However, we know that $S_{u,v}$ induces an intersection graph of intervals corresponding to the vertices $\{w \colon w \in V, l_u \leq l_w <r_w \leq r_v\} \cup \{u,v\}$, and $G'$ contain both $u$ and $v$. So, $N[z] \cap G' \neq \phi$, $N[z'] \cap G' \neq \phi$. Therefore $V(G')\cup \{z,z'\}$ induces a balanced connected subgraph in $G$. It contradicts our assumption. Thereby, we conclude the proof. 

\paragraph{Time Complexity:} We basically use the select Steiner algorithm (Algorithm~\ref{algo:St-int}) as a subroutine that we call for every pair of intervals. Moreover the graph can be stored in a range tree and, for any pair of vertices $u$ and $v$, the set $S_{u,v}$ can be obtained in $\mathcal{O}(\log n)$  time. We have shown that the computing Steiner takes $\mathcal{O}(n^2\log n)$ time, and thus our algorithm computes a \bcs~in $\mathcal{O}(n^4\log n)$ time.

\begin{theorem}
Let $G$ be an interval graph whose $n$ vertices are colored either red or blue.
    Then \bcs~problem on $G$ can be solved in $\mathcal{O}(n^4 \log n)$ time.
\end{theorem}

	



\old{
 Below in \ref{algo_1.1}, we describe an algorithm to find maximum balanced connected subgraph in an interval graph. 
 
 \IncMargin{1em}
 \begin{algorithm}[H]
 	\SetKwData{Left}{left}\SetKwData{This}{this}\SetKwData{Up}{up}
 	\SetKwFunction{Union}{Union}\SetKwFunction{FindCompress}{FindCompress}
 	\SetKwInOut{Input}{Input}\SetKwInOut{Output}{Output}
 	\Input{A bicolored interval graph $ G=(V,E)$ with $ V=\{v_1, v_2, \dots v_n\} $.}
 	\Output{Maximum sized \bcs~in $G$. }
 	\BlankLine
 	Find an interval representation $ I_G $ of $ G $.\\
 	Sort the intervals of $ I_G $ according to the left endpoints. We assume that, this ordering is $I_1 < I_2< \dots < I_n$, where for each $ i, $ $ I_i $ corresponds the vertex $ v_i $,  $1 \leq i \leq n$.
 	
 	\For{$i\leftarrow 1$ \KwTo $n$}
 	{
 		\For{$j\leftarrow (i+1)$ \KwTo $n$}{
 			$ a= l(I_i) $ \tcp*{$ l(I_i) $ denotes the left endpoint of $ I_i $}
 			$ c= $ max $\{r(I_i),r( I_j)\}$ \tcp*{$ r(I_i) $ and $r(I_j)$ denote the right endpoint of $ I_i $ and $ I_j $, respectively.} 
 			$ v_c $= the vertex having right endpoint at $ c $.
 		$ S_{i,j}=\{w \colon w \in V, a \leq l(I_w) <r(I_w) \leq b \} $\tcp*{ interval $ I_w$ corresponds the vertex $ w $} 
 		$ H_{i,j}= G[S_{i,j}] $ \tcp*{ For $ U \subseteq V $, $ G[U] $ denotes the subgraph of $ G $ induced by $ U $} 
 		\If{$H_{i,j}$ is connected}{$ T_{i,j} =$ output of Compute$\_$\bcs$\_$interval$\_$graph$ (H_{i,j},v_i,v_c) $}
 		\Else{$ T_{i,j} = \phi$}
 	}
 	}
 Return maximum sized $ T_{i,j}, 1 \leq i < j \leq n$
 	\caption{ Compute$\_$\bcs$\_$interval$\_$graph$ (G) $}\label{algo_1.1}
 \end{algorithm}\DecMargin{1em}	

\IncMargin{1em}
\begin{algorithm}[H]
	\SetKwData{Left}{left}\SetKwData{This}{this}\SetKwData{Up}{up}
	\SetKwFunction{Union}{Union}\SetKwFunction{FindCompress}{FindCompress}
	\SetKwInOut{Input}{Input}\SetKwInOut{Output}{Output}
	\Input{(i) A connected bicolored interval graph $ H=(V,E)$ with an ordered set of $ n $ vertices. Assume that $ r $ and $ b $ are the number of red and blue vertices respectively, in $ V $ with $ r \leq b $ and \\ (ii) two vertices $ u,v \in V $}. 
	\Output{A \bcs~in $ H $ with size $ 2r $ containing $ u $ and $ v $ (if exists).}
	\BlankLine
	$ K $= set of all red vertices in $ V $.\\
	$ T= K \cup \{u,v\} $\\
	$T' $=output of Compute$\_$Steiner$\_$tree$\_$interval$\_$graph$ (H, T) $.\\
	$ b'= $ number of blue vertices in $ T' $\\
	
	\If{$ b'>r $}{return $ \phi $}
	\If{$ b'=r $}{return $T' $}
	\If{$ b'<r $}{return $ G[V(T') \cup X] $, where $ X \subset V$ is the set of $ (r-b') $ blue vertices with $ X \cap V(T')= \phi$.  \tcp*{ $ X $ always exists, as $ r<b $} }

	\caption{ Compute$\_$\bcs$\_$interval$\_$graph$ (H,u,v) $}\label{algo_1.2}
\end{algorithm}\DecMargin{1em}	

}

\subsection{Circular-Arc Graphs}
We study the \bcs~problem on circular-arc graphs. 
We are given a bicolored the circular arc graph $G=(V_R\cup V_B, E)$, where the set $V_R$ and $V_B$ contains a set of red and blue arcs, respectively. With out loss of generality we assume that the given input arcs fully cover the circle. Otherwise it  becomes an interval graph and we use the algorithm of the interval graph to get an optimal solution. 

Let us assume that $H$ be a resulting maximum balanced connected subgraph  of $G$, and
let $S$ denote the set of vertices in $H$. 
Since $H$ is a connected subgraph of $G$, $H$ covers the circle either partially or entirely. 
We propose an algorithm that computes a maximum size balanced connected subgraph $H$ of $G$ in polynomial time. Without loss of generality we assume that $V_B\leq V_R$. For any arc $u\in V$, let $l(u)$ and $r(u)$ denote the two endpoints of $u$ in the clockwise order of the endpoints of the arc. To design the algorithm, we shall concentrate on the following two cases -- 
\textbf{Case~A} and \textbf{Case~B}. 
In Case~A, we check all possible instances while the the output set does not cover the circle fully. Then, in Case~B, we handle all those instances while the output covers the entire circle. Later, we prove that the optimum solution lies in one of these instances. The objective is to reach the optimum solution by exploiting these instances exhaustively. 
\vspace{.2cm}

\noindent \textbf{Case A: $H$ covers the circle partially:} In this case, there must be a clique of arcs $K$ ($|K|\ge 1$) that is not present in the optimal solution. 
We consider any pair of arcs $u,v\in V$ such that $r(u)\prec l(v)$ in the clockwise order of the endpoints of the arcs, and consider the vertex set 
$S_{u,v}= \{w \colon w \in V, l_v \leq l_w <r_w \leq r_u\} \cup \{u,v\}$. Then, we use the Algorithm~\ref{algo:bcs-int} to compute maximum \bcs~on $G[S_{u,v}]$. This process is repeated for each such pair of arcs, and report the \bcs~with maximum number of arcs.
    
\vspace{.2cm}

\noindent\textbf{Case B: $H$ covers the circle entirely:}  In this case, $S$ must contains $2|V_B|$ number of arcs and in fact that is the maximum number of arcs $S$ can opt. 
In order to compute such a set $S$, first we add the vertices in $V_B$ to $S$, then consider the vertices in $V_B$ as a set $T$ of terminal arcs and we need to find a minimum number of red arcs $D\in V_R$ to span over $T$. We further distinguish between two cases. 

\begin{description}
   \item[B.1. $T\cup D$ covers the circle partially] 
   Clearly, this case is similar to Case~A without some extra red arcs that would be added afterwards to  ensure that $S$ contains $ 2|V_B|$ arcs. Similar to Case~A, we again try all possible subsets obtained by pair of vertices $u,v$ with $r(u)\prec l(v)$ and $S_{u,v}$ contains all blue vertices and we find optimal Steiner tree by using Procedure~\ref{algo:St-int}. Then, we add $(|V_R|-|D|)$ (where $D$ is the set of Steiner arcs) many red arcs from $V_R$ to $S$.      
  
    \item[B.2. $T\cup D$ covers the circle entirely] 
    First, we obtain a set $\mathcal{C}$ of $m$ (for some $m\in [n]$) components from $T$. We may see each component $C \in \mathcal{C}$ as an arc and the neighborhood set $N(C)$ as the union of the neighborhoods of the arcs contained in $C$. Observe that, for any component  $C\in \mathcal{C}$, $D$ contains either one arc from $N(C)$ that covers $C$, or two arcs from $N(C)$ where none of them individually covers $C$. 
    Let us consider one component $C  \in \mathcal{C}$. Let $l(C)$ and $r(C)$ be the left and right end points of $C$, respectively. If $|N(C)\cap D|=1$, we consider each arc from $N(C)$ separately that contains $C$.    For each such arc $I(C)\in N(C)$, we do the following three step operations --1) include $I(C)$ in $D$, 2) remove $N(C)$ from the graph, 3) include two blue arcs $(l(I(C)),r(C))$ and $(r(C),r(I(C)))$ in the vertex set of the graph. 
    Now, $T=T\cup \{[l(I(C)),r(C)), (r(C),r(I(C))]\}$. 
    We give this processed graph that is an interval graph, 
    as an input of the Steiner tree (Procedure~\ref{algo:St-int}) and look for a tree with at most $(|D|-1)$ Steiner red arcs.     Else, when $|N(C)\cap D|=2$, we take the arcs $C_\ell$ and $C_r$ from $N(C)$ with leftmost and rightmost endpoints, respectively, in $D$. We do the same three step operations --1) include $C_\ell$ and $C_r$ in $D$, 2) remove $N(C)$ from the graph, 3) include two blue arcs $(l(C_\ell), l(C)))$, $(r(C), r(C_r)))$.    
    Now, $T=T\cup \{[l(C_\ell), l(C))), (r(C), r(C_r))]\}$. 
    We give this processed graph that is an interval graph, 
    as an input of the Steiner tree (Procedure~\ref{algo:St-int}) and look for a tree with at most $(|D|-2)$ Steiner red arcs.
    This completes the procedure. 
\end{description}   

\paragraph{Correctness:} We prove that our algorithm yields an optimum solution. The proof of correctness follows from the way we have designed the algorithm. The algorithm is divided into two cases. For case~A, the primary objective is to construct the instances from a circular-arc graph to some interval graph. Thereafter, we can solve it optimally. Now, Case~B is further divided into two sub-cases. Here we know that all blue vertices present in optimum solution. Therefore, our goal is to employ the Steiner tree algorithm with terminal set $T=V_B$.  Note, for B~1 we again try all possible subsets obtained by pair of vertices $u,v$ where $r(u)\prec l(v)$ and $S_{u,v}$ contains all blue vertices. Note, $G[S_{u,v}]$ is an interval graph and $V_B$ is the set of terminals (since we assumed, w.l.o.g, $V_B\leq V_R$). Therefore we directly apply the Steiner tree procedure and obtain the optimum subset for each such pair. Indeed, this process reports the maximum \bcs. In the case B~2 we modify the input graph in three step operations. Moreover, we update the expected output size to make it coherent 
to the modified input instance. This process is done for one arbitrary
component of $T$ (of size $\ge 1$), which gives an interval graph. 
Clearly, the choice of this component makes no impact on the size of \bcs. 
Thereafter, we follow the Steiner tree algorithm on this graph. 
Moreover, the algorithm exploits all possible cases and reduce the graph into interval graph without affecting the size of the \bcs. Thereby, putting everything together, we conclude the proof.

\paragraph{Time Complexity:} The algorithm is divided into two cases. For case~A, we try all pair of arcs that holds certain condition and consider the subset (note, such subset can be computed in $\mathcal{O}(\log n)$ time given the clockwise order of the vertex set and a range tree where the arcs are stored). For each such subset we use the algorithm for interval graph to compute the maximum \bcs. This whole process takes $\mathcal{O}(n^6\log n)$ time. Now, in case~B.1., we do the same procedure, but, directly use the Procedure~\ref{algo:St-int} which saves $\mathcal{O}(n)$ time. In case~B.2., for one terminal component, we construct $\mathcal{O}(n)$ many interval graphs and apply Steiner tree algorithm directly on them to compute the maximum \bcs. Moreover, the complexity of Case~A dominates and we get the total running time $\mathcal{O}(n^6\log n)$.

\begin{theorem}
    Given an $n$ vertex circular-arc graph $G$, where the vertices in $G$ are colored either red or blue, the \bcs~problem on $G$ can be solved in $\mathcal{O}(n^6 \log n)$ time.
\end{theorem}

\subsection{Permutation Graphs}
In this section, we study the \bcs~problem on permutation graphs. 
A graph $G=(V,E)$ with $|V|=n$ is called a permutation graph if there exists a bisection 
$f \colon V \rightarrow \{1,2,\ldots,n\}$ and a permutation $\pi$ of order $n$ such that for every pair of vertices $u,v \in  V, (u,v)\in E \Leftrightarrow (f (u)-f (v)) (\pi(f (u)) -\pi(f (v))) < 0$. 
This can be represented as an intersection graph of line segments whose endpoints lie on parallel lines $ y=0 $ and $ y=1 $. We order the vertices of $G$ based on the endpoints of their corresponding lines on  $y=1$. Let $v_1 < v_2< \ldots < v_n$ be the ordering, 
where $v_i < v_j \Leftrightarrow p_i < p_j $ (Assuming $(p_i,1)$ and $(p_j,1)$ are the endpoints of the lines corresponding to vertices $v_i$ and $v_j$, respectively). For each pair of vertices $ v_i,v_j \in V $ (where $1 \leq i < j \leq n$), we define the subgraph $ G_{i,j} $  induced by $\{v_{i}, v_{i+1}, \dots,v_{j-1}, v_j\}$ in $G$. Let $ r_{i,j} $ and $ b_{i,j} $ denote the set of red vertices and  blue vertices in $G_{ij}$, respectively.

We propose an algorithm to compute 
maximum \mmbcs~on permutation graphs.   
We search for all pair of numbers $ i $ and $ j $ for which $ G_{i,j} $ is connected and $ G_{i,j} $  has  a \bcs~with size $ \min  \{2|r_{i,j}|, 2|b_{i,j}|\} $. Finally, we report a \bcs~having the maximum size. For each connected subgraph $G_{i,j}$, to get  a \bcs~(if exists) with size $\min  \{2|r_{i,j}|, 2|b_{i,j}|\} $,  we apply the algorithm for the Steiner tree 
problem on permutation graphs with terminal set 
$b_{i,j}$ (if $ |b_{i,j}| \leq |r_{i,j}|$) or  
$r_{i,j}$ (if $ |r_{i,j}| < |b_{i,j}|$).  We use the following theorem to solve Steiner tree problem in $G_{i,j}$.

\begin{theorem} \cite{colbourn1990permutation}
	A minimum cardinality Steiner tree in an $ n $ vertex permutation graph can be found in $\mathcal{O}(n^3)$ time.
\end{theorem}

Now, let $T_{i,j}$ be the solution of Steiner tree problem in $ G_{i,j} $ with terminal vertices $ b_{i,j}$, where $ |b_{i,j}| \leq |r_{i,j}|$. 
Recall that, we only consider the cases where $G_{i,j}$ is connected. If the number of red vertices in $T_{i,j}$ is less than $2|b_{i,j}|$, 
then we obtain a \bcs~of size $2|b_{i,j}|$ 
by simply adding the required number of red vertices.
Note, this does not affect the connectivity. 
By using this method we find a \bcs~with size  $\min \{2|r_{i,j}|, 2|b_{i,j}|\} $ in 
$G_{i,j}$ for $1 \leq i < j \leq n $.

\paragraph{Correctness:}
We prove that our algorithm yields an optimum solution.
Let $H$ be an optimal solution of size $t$. Let $H$ consists of vertices $v_{k_1}, v_{k_2}, \dots, v_{k_t}$, where $k_1 < k_2 , \dots, k_t$. $G_{k_1,k_t}$ must be connected (otherwise, $H$ becomes disconnected). So we apply Steiner tree problem in  $G_{k_1,k_t}$. Now, we show that $|V(H)|= \min  \{2|r_{i,j}|, 2|b_{i,j}|\} $. Let  $|V(H)| \neq  \min  \{2|r_{i,j}|, 2|b_{i,j}|\} $ then there exists at least one blue vertex $u$ and one red vertex $v$, such that $u,v \in G_{i,j} \setminus H$. As $G_{k_1,k_t}$ is an intersection graph of line segments corresponding to the vertices $v_{k_1}, v_{k_2}, \dots, v_{k_t}$, where $k_1 < k_2< \dots< k_t$, and $H$ contain both $v_{k_1}$ and $v_{k_t}$. So, $N[u] \cap H \neq \phi$ and $N[v] \cap H \neq \phi$. Therefore $V(H) \cup \{u,v\}$ induces a balanced connected subgraph in $G$. It contradicts to the fact that $H$ is a maximum induced balanced connected subgraph in $G$. Thereby, we conclude the proof.

\paragraph{Time Complexity:} We use the algorithm \cite{colbourn1990permutation} to find minimum cardinality Steiner tree as a subroutine that we call for every pair of integers $i,j$ where $1 \le i \le j \le n$ and $G_{i,j}$ is connected. In case of permutation graphs, we may require linear time to obtain an induced subgraph $G_{i,j}$ for any pair of vertices $u,v\in V$.
Now, computing Steiner on permutation graphs takes $\mathcal{O}(n^3)$ time, and 
thus our algorithm computes a \bcs~in $O(n^6)$ time.

\begin{theorem}
	Given an $n$ vertex permutation graph $G$, where the vertices in $G$ are colored either red or blue, the \bcs~problem on $G$ can be solved in $\mathcal{O}(n^6)$ time.
\end{theorem}

\subsection{FPT Algorithm}\label{fpt}

In this section, we show that  the \bcs~problem is fixed-parameter tractable for general graphs while parameterized by the solution size. 
Let $G=(V,E)$ be a simple connected graph, 
and let $k$ be a given parameter. 
A family $\mathcal{F}$ of functions from $V$ to $\{1,2, \dots, k\}$ 
is called a \colg{perfect hash family} for $V$ if the following condition holds. For any subset $U \subseteq V$ with $|U|=k$, 
there is a function $f \in \mathcal{F}$ which is injective on $U$, i.e., $ f|_U $ is one-to-one. For any graph of $ n $ vertices and a positive integer $ k $, 
it is possible to obtain a perfect hash family of size $2^{\mathcal{O}(k)} \log n$ in
$2^{\mathcal{O}(k)} n\log n$ time; see \cite{alon1995color}. 
Now, the $k$-\bcs~problem can be defined as follows. 

\begin{tcolorbox}[colback=red!5!white]
	\noindent {\bf{\color{red!50!black} {$k$-\mmbcs~Problem ($k$-\bcs)}}}\\
	{\bf Input:} A graph $G=(V,E)$, with node set $V=V_R\cup V_B$ partitioned into red nodes ($V_R$) and blue nodes ($V_B$), and a positive integer $k$.\\
	{\bf Output:} Balanced connected subgraph of size $k$.
\end{tcolorbox}

We employ two methods to solve the $k$-\bcs~problem: (i) \colg{color coding technique}, and (ii) \colg{batch procedure}. Our approach is motivated by the approach of Fellow et al.~\cite{Fellows2011}, where they have used these techniques to provide a FPT-algorithm for the graph motif problem. Suppose $H$ is a solution to the $k$-\bcs~problem and  $\mathcal{F}$ is a perfect hash family for $V$. This ensures us that at least one function of $\mathcal{F}$ assigns vertices of $H$ 
with $k$ distinct labels. Therefore, if we iterate through all functions of $\mathcal{F}$ and find the 
subsets of $V$ of size $k$ that are distinctly labeled by our hashing, we are guaranteed to find $H$.
Now, we try to solve the following problem: 
Given a hash function $f \colon V \rightarrow \{1,2, \dots, k\}$ from perfect hash family $\mathcal{F}$, decide if there is a subset $ U \subseteq V $ with  $|U|= k $ such that $ G[U] $ is a  balanced connected subgraph of $ G $ and  $ f|_U $ is one-to-one.

First, we create a table, denoted by $M$. 
For a label $L \subseteq \{1,2,\dots,k\}$ 
and a color-pair $(b,r)$ of non-negative integers where $b+r=|L|$, 
we put $ M[v~ ;~ L,~ (b,r)] =1$ if and only if  there exists a subset $ U\subseteq V $ of vertices such that following conditions holds: 
  \begin{itemize}
    \item[(i)] $v \in U $,
    \item[(ii)] $f|_U= L$,
    \item[(iii)] $ G[U] $ is connected,
    \item[(iv)] $U$ consisting exactly $b$ blue vertices and $ r $ red vertices.
\end{itemize}

Notice that, the total number of entries of this table is $\mathcal{O}(2^k kn)$. If we can fill all the entries of the table $ M $, then we can just look at the entries $ M[v~ ; ~\{1,2,\dots, k\}, ~(\frac{k}{2},\frac{k}{2})] $, $ \forall v \in V $, and if any of them is one then we can claim that the $ k$-\bcs~problem has a solution. Now we use the batch procedure to compute $ M[v~ ;~ L,~ (b,r)] $ for each subset $ L \subseteq \{1, 2,\dots k\} $, and for each pair $ (b,r) $ of non-negative integers such that $ (b+r)=|L|$. Now, we explain the batch procedure. Without loss of generality we assume that $ L= \{1, 2, \dots, t\} $, $ f(v)=t $, and the color of $ v $ is red. 

\vspace{.2cm}

\noindent \textbf{\underline{Batch Procedure ($ v, L, (b,r) $):}}
 
 \begin{description}
 	\item[(1)] \textbf{Initialize:} Construct the set $ \mathcal{S} $ of pairs $ (L', (b',r')), $ where $ b'+r'=|L'| $ such that $ L' \subseteq \{1, 2,\dots, t-1\} $, ~$ b' \leq b,~r' \leq  r-1$ and $ M [u~ ;~ L',~ (b',r')] =1$ for some neighbour $ u $ of $ v $.
 	
 	\item[(2)] \textbf{Update:} If there exists two pairs $\{(L_1, (b_1,r_1)), (L_2, (b_2,r_2))\}  \in \mathcal{S} $ such that 	$ L_1 \cap L_2 = \phi $ and $ (b_1, r_1) + (b_2, r_2) \leq  (b, r-1)$, then add $(L_1 \cup L_2, (b_1+b_2, r_1+r_2)) $ into $ \mathcal{S} $. Repeat this step until unable to add any more.
 	
 	\item[(3)] \textbf{Decide:} Set $ M[v~ ;~ L,~ (b,r)] =1$ if $ (\{1, 2, \dots, t-1\}, (b, r-1)) \in \mathcal{S}$, else 0.

 \end{description}

\begin{lemma}
	
	The batch procedure correctly computes  $ M[v~ ;~ L,~ (b,r)]$ for all $ v $, $ L \subseteq \{1, 2, \dots, k \} $ with $ b+r= |L| $.
	
	\end{lemma}

\begin{proof}
		Without loss of generality we assume that $ L= \{1, 2, \dots, t\}, $ $ f(v)=t $ and color of $ v $ is red. We have to show that $ M[v~ ;~ L,~ (b,r)]=1 \Leftrightarrow$ there exists a connected subgraph, containing $ v $, and having exactly $ b $ blue vertices and $ r $ red vertices. Firstly, we assume that  $ M[v~ ;~ L,~ (b,r)]=1$.  So, $ (\{1, 2, \dots, t-1\}, (b, r-1)) \in \mathcal{S}$. 
		So, there must exist some neighbours $ \{v'_1, v'_2, \dots ,v'_l\} $ of $v$ such that $ M[v'_1 ;~ L_1,~ (b_1,r_1)]=M[v'_2 ;~ L_2,~ (b_2,r_2)]=\dots= M[v'_l ;~ L_l,~ (b_l,r_l)]=1$ with $\bigcup\limits_{i=1}^{l} L_{i}=L \setminus \{t\}, \sum_{i=1}^{l}b_i=b, \sum_{i=1}^{l}r_i=r-1$ where $ L_1, L_2, \dots, L_l $ are pairwise disjoint. Thus, there exists a connected subgraph containing $ \{v, v'_1, \dots, v'_l\} $ having exactly $ b $ blue vertices and $ r $ red vertices.
	Other direction of the proof follows from the same idea.
	\end{proof}

\begin{lemma}
	
Given a hash function $ f \colon V \rightarrow \{1,2, \dots , k\} $, batch procedure fill all the entries of table $ M $ in $ \mathcal{O}(2^{4k}k^3 n^2) $ time.
	
\end{lemma}

\begin{proof}
	In the first step of the batch procedure, the initialization depends on the number of the search process in the entries correspond to the neighbour of $v$. 
	Now, number of the search process is bounded by the size of $M$. 
	The first step takes $\mathcal{O}( 2^{k}kn)$ time. Now the size of $ \mathcal{S} $ can be at most $ 2^{k}k$. Each update takes 
	$\mathcal{O}(2^{2k}k^2)$ time. So step~$2$ 
	takes $\mathcal{O}(2^{3k}k^3)$ time. Step~$1$ and $2$ together takes $\mathcal{O}(2^{3k}k^3+ 2^kkn)= \mathcal{O}(2^{3k}k^2n) $ time. After the first two steps, the value of $ M[v~ ;~ L,~ (b,r)]$ can be decided in $\mathcal{O}(2^kk)$ time. These three steps are independent in terms of the running time. As the number of entries in $M$ is $\mathcal{O}(2^kkn)$.
	Hence, the total running time to fill all the entries in $ M $ is $ \mathcal{O}(2^{4k}k^3 n^2) $. \qed
\end{proof}

\noindent Now, our algorithm for the $ k $-\bcs~problem is the following:
\begin{description}
	\item[1.]Construct a perfect hash family $ \mathcal{F} $ of $ 2^{\mathcal{O}(k)}\log n$ hash functions in  $ 2^{\mathcal{O}(k)} n\log n $ time.
	
	 \item[2.] For each function, $f \in \mathcal{F}$ build the table $ M $ using batch procedure. For each function $f\in \mathcal{F}$ it takes $ \mathcal{O}(2^{4k}k^3 n^2) $ time.
	 
	 \item[3.] As each $f\in \mathcal{F}$ is perfect, by an exhaustive search through all function in $ \mathcal{F} $ our algorithm correctly decide whether there exists a balanced connected subgraph of $k$ vertices.
	 We output yes, if and only if there is a vertex $ v $ and $ f \in \mathcal{F}$ for which the corresponding table $ M $ contains the entry one in $  M[v~ ; ~\{1,2,\dots, k\}, ~(\frac{k}{2},\frac{k}{2})] $.
\end{description}

\begin{theorem}
The $k$-\bcs~problem can be solved in time 
$2^{\mathcal{O}(k)}n^2 \log n$ time.
    
\end{theorem}



{\old{\section{Conclusion}
In this paper, we have studied the tractability landscape of the maximum balanced connected subgraph problem on geometric intersection graphs. We gave several algorithmic and \np-hardness results. 
We show that \bcs~problem remains \np-hard on restricted geometric intersection graph classes such as unit-disk graphs, outer-string graphs, complete grid graphs, and unit-square graphs. 
We design polynomial-time algorithms for the interval graphs ($\mathcal{O}(n^4\log n)$ time), circular-arc graphs ($\mathcal{O}(n^6\log n)$ time) 
and permutation graphs $\mathcal{O}(n^6)$. Finally, we have shown that the \bcs~problem is fixed-parameter tractable for general graphs ($ 2^{\mathcal{O}(k)}n^2 \log n $) while parameterized by number of vertices in a balanced connected subgraph.

 An intriguing open problem in this area is to design a constant-factor approximation algorithm for the \bcs~problem on unit disk graphs. Due to the nature of the problem, we have identified that the general geometric approximation schemes such as local search or shifting strategy do not work directly. Nonetheless, the \bcs~is one of those problems that does not respect the local neighborhood and it is difficult to predict anything without seeing the distribution of the red and blue vertices in the global solution. Therefore, it is indeed challenging to get a tight lower bound, and then design a strategy that does not deviate much from that bound.
 }}



\section*{Acknowledgement} 

The authors would like to thank Joseph S. B. Mitchell for interesting discussions during the early stages of the research.

\bibliographystyle{splncs04}
\bibliography{samplepaper}

\begin{thebibliography}{10}
\providecommand{\url}[1]{\texttt{#1}}
\providecommand{\urlprefix}{URL }
\providecommand{\doi}[1]{https://doi.org/#1}

\bibitem{alon1995color}
Alon, N., Yuster, R., Zwick, U.: Color-coding. J. ACM  \textbf{42}(4),
  844--856 (1995)

\bibitem{bhore2019balanced}
Bhore, S., Chakraborty, S., Jana, S., Mitchell, J.S., Pandit, S., Roy, S.: The
  balanced connected subgraph problem. In: Conference on Algorithms and
  Discrete Applied Mathematics. pp. 201--215. Springer (2019)

\bibitem{bodlaender2007quadratic}
Bodlaender, H.L., Fellows, M.R., Langston, M.A., Ragan, M.A., Rosamond, F.A.,
  Weyer, M.: Quadratic kernelization for convex recoloring of trees. In:
  International Computing and Combinatorics Conference. pp. 86--96. Springer
  (2007)

\bibitem{bodlaender1992two}
Bodlaender, H.L., Fellows, M.R., Warnow, T.J.: Two strikes against perfect
  phylogeny. In: International Colloquium on Automata, Languages, and
  Programming. pp. 273--283. Springer (1992)

\bibitem{bodlaender1995intervalizing}
Bodlaender, H.L., de~Fluiter, B.: Intervalizing k-colored graphs. In:
  International Colloquium on Automata, Languages, and Programming. pp. 87--98.
  Springer (1995)

\bibitem{Bonnet2017}
Bonnet, {\'{E}}., Sikora, F.: The graph motif problem parameterized by the
  structure of the input graph. Discrete Applied Mathematics  \textbf{231},
  78--94 (2017)

\bibitem{clark1991unit}
Clark, B.N., Colbourn, C.J., Johnson, D.S.: Unit disk graphs. In: Annals of
  Discrete Mathematics, vol.~48, pp. 165--177. Elsevier (1991)

\bibitem{colbourn1990permutation}
Colbourn, C.J., Stewart, L.K.: Permutation graphs: connected domination and
  steiner trees. Discrete Mathematics  \textbf{86}(1-3),  179--189 (1990)

\bibitem{dehmer2010structural}
Dehmer, M.: Structural analysis of complex networks. Springer Science \&
  Business Media (2010)

\bibitem{Fellows2011}
Fellows, M.R., Fertin, G., Hermelin, D., Vialette, S.: Upper and lower bounds
  for finding connected motifs in vertex-colored graphs. J. Comput. Syst. Sci.
  \textbf{77}(4),  799--811 (2011)

\bibitem{fellows1993dna}
Fellows, M.R., Hallett, M.T., Wareham, H.T.: Dna physical mapping: Three ways
  difficult. In: European Symposium on Algorithms. pp. 157--168. Springer
  (1993)

\bibitem{garey1977rectilinear}
Garey, M.R., Johnson, D.S.: The rectilinear steiner tree problem is
  np-complete. SIAM Journal on Applied Mathematics  \textbf{32}(4),  826--834
  (1977)

\bibitem{imai1983finding}
Imai, H., Asano, T.: Finding the connected components and a maximum clique of
  an intersection graph of rectangles in the plane. Journal of algorithms
  \textbf{4}(4),  310--323 (1983)

\bibitem{karp1972reducibility}
Karp, R.M.: Reducibility among combinatorial problems. In: Complexity of
  computer computations, pp. 85--103. Springer (1972)

\bibitem{kikuno1980np}
Kikuno, T., Yoshida, N., Kakuda, Y.: The np-completeness of the dominating set
  problem in cubic planer graphs. IEICI Transactions (1976-1990)
  \textbf{63}(6),  443--444 (1980)

\bibitem{Lacroix2016}
Lacroix, V., Fernandes, C.G., Sagot, M.: Motif search in graphs: Application to
  metabolic networks. IEEE/ACM Transactions on Computational Biology and
  Bioinformatics  \textbf{3}(4),  360--368 (2006)

\bibitem{lacroix2006motif}
Lacroix, V., Fernandes, C.G., Sagot, M.F.: Motif search in graphs: application
  to metabolic networks. IEEE/ACM Transactions on Computational Biology and
  Bioinformatics (TCBB)  \textbf{3}(4),  360--368 (2006)

\bibitem{pecher2010clique}
P{\^e}cher, A., Wagler, A.K.: Clique and chromatic number of circular-perfect
  graphs. Electronic Notes in Discrete Mathematics  \textbf{36},  199--206
  (2010)

\end{thebibliography}

\end{document}